\documentclass[conference]{IEEEtran}
\IEEEoverridecommandlockouts
\usepackage{cite}
\usepackage{amsmath,amssymb,amsfonts}
\usepackage{algorithm, algorithmic}
\usepackage{graphicx}
\usepackage{textcomp}
\usepackage{xcolor}
\def\BibTeX{{\rm B\kern-.05em{\sc i\kern-.025em b}\kern-.08em
    T\kern-.1667em\lower.7ex\hbox{E}\kern-.125emX}}
\DeclareMathOperator{\E}{\mathbb{E}}

\newtheorem{theorem}{\textbf{Theorem}}
\newenvironment{proof}{\noindent\textit{Proof}: }{\hfill$\square$\par}

\begin{document}
\title{Online V2X Scheduling for Raw-Level \\ Cooperative Perception
} 

\author{\IEEEauthorblockN{Yukuan Jia, Ruiqing Mao, Yuxuan Sun, Sheng Zhou, Zhisheng Niu}
\IEEEauthorblockA{Beijing National Research Center for Information Science and Technology\\
Center for Intelligent Connected vehicles and Transportation, Tsinghua University\\
Department of Electronic Engineering, Tsinghua University, Beijing, P.R. China\\
Email: \{jyk20@mails., mrq20@mails., sunyuxuan@, sheng.zhou@, niuzhs@\}tsinghua.edu.cn}
\thanks{This work is sponsored in part by the National Key R\&D Program of China No. 2020YFB1806605, by the Nature Science Foundation of China (No. 62022049, No. 61871254, No. 61861136003), by the China Postdoctoral Science Foundation No. 2020M680558, and by the project of Tsinghua University-Toyota Joint Research Center for AI Technology of Automated Vehicle (No. TTAD2021-08). }
}

\maketitle

\begin{abstract}
Cooperative perception of connected vehicles comes to the rescue when the field of view restricts stand-alone intelligence. 
While raw-level cooperative perception preserves most information to guarantee accuracy, it is demanding in communication bandwidth and computation power. 
Therefore, it is important to schedule the most beneficial vehicle to share its sensor in terms of supplementary view and stable network connection. 
In this paper, we present a model of raw-level cooperative perception and formulate the energy minimization problem of sensor sharing scheduling as a variant of the Multi-Armed Bandit (MAB) problem. 
Specifically, volatility of the neighboring vehicles, heterogeneity of V2X channels, and the time-varying traffic context are taken into consideration. 
Then we propose an online learning-based algorithm with logarithmic performance loss, achieving a decent trade-off between exploration and exploitation. 
Simulation results under different scenarios indicate that the proposed algorithm quickly learns to schedule the optimal cooperative vehicle and saves more energy as compared to baseline algorithms.
\end{abstract}


\section{Introduction}
In the area of autonomous driving, the stand-alone intelligence has fundamental limitations due to single perception viewpoint. 
For example, the autonomous vehicle is unable to detect a pedestrian occluded by other vehicles or in the blind zone, which causes potential danger.
Thanks to the Vehicle-to-Everything (V2X) network\cite{zeadally2020vehicular}, \textit{connected vehicles} will be a promising solution to occlusions and blind zones, by aggregating sensor data collected from different views. 
The vehicles can exchange their sensor data and \textit{cooperatively} perceive the targets, thus improving the reliability of perception\cite{yang2021machine}. 

There are three levels of \textit{cooperative perception} (CP), a.k.a. collective perception, which differs in how the data from multiple sources are combined. 
As one of the pioneering approaches to raw-level CP, Cooper \cite{chen2019cooper} realizes cooperative perception by exchanging raw 3D point cloud data, outperforming individual perception with extended sensing area and improved detection confidence.
For feature-level, Graph Neural Network (GNN) is leveraged in \cite{wang2020v2vnet} to aggregate lightweight feature data from multiple nearby vehicles, maintaining high accuracy while reducing communication bandwidth requirements. 
Object-level CP, in which only detection results are exchanged, is demonstrated in \cite{kim2015impact}, including see-through forward collision warning, overtaking, and lane-changing assistance scenarios. 
Ref. \cite{higuchi2019value} anticipates the value of information, and schedules only the most important data to avoid network congestion. 
In addition, a standard of object-level CP has been released by ETSI\cite{etsi}.

Among the three types of CP, raw-level CP has the advantage of preserving all the information so as to better exploit the additional view to the full potential. 
Therefore, we stick to the raw-level and also incorporate feature-level as a patch to reduce communication cost. 
Note that when there are multiple candidates to share their sensor data, it is challenging to predict how one view would assist the detection of another, due to unclear occlusion relationship, heterogeneous sensor qualities, as well as the black-box feature of deep neural networks.
To address this issue, we adopt a learning-based framework \cite{sun2019adaptive}. 
By trial and error, the ego vehicle gradually learns which vehicle can provide a superior supplementary view while the wireless connection is satisfactory at the same time. 

The computation system for autonomous driving could take over a thousand watts, reducing the mileage of a vehicle by up to 30\%\cite{liu2020computing}. 
In fact, the computation load varies significantly with different AI models. 
For example, YOLOX-X\cite{ge2021yolox} costs $282$ giga floating-point operations (GFLOPs) while its lightweight version, YOLOX-Tiny, costs only $6.45$ GFLOPs. 
The detection performance is affected by the model size, as shown in Fig. \ref{yolox}. 
In order to be energy-efficient, our AI model for the cooperative perception task is dynamic\cite{han2021dynamic}.
A policy network decides which layer or block in the model to activate given the input data from multiple sources.
Consequently, the computation load of power-consuming AI models can be reduced, achieving energy-efficient environmental perception.

\begin{figure}[!t]
	\centering
	\includegraphics[width=7.3cm]{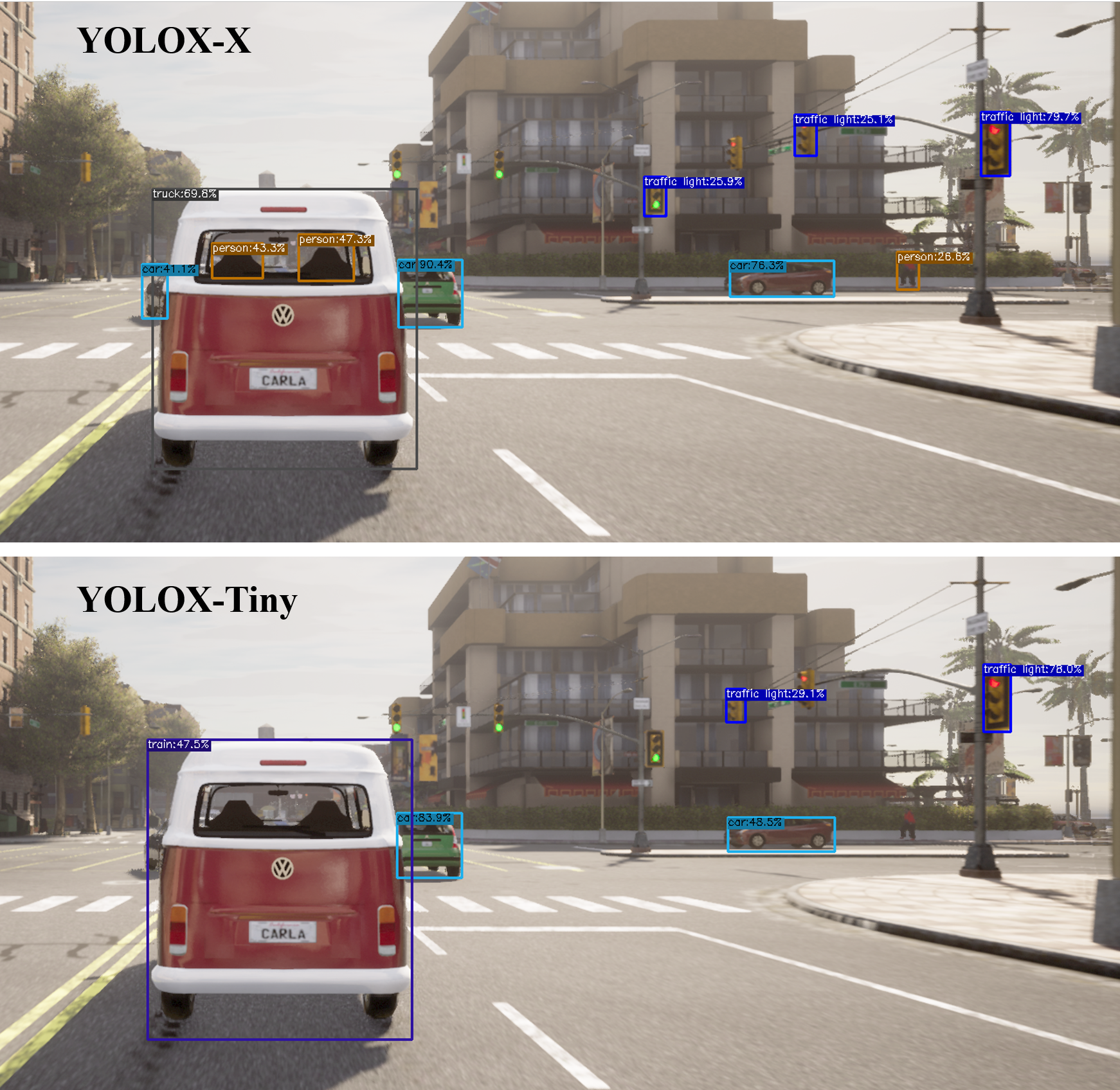}  
	\caption{Comparison of a large model, YOLOX-X\cite{ge2021yolox}, and a lightweight model, YOLOX-Tiny, on a CARLA-generated \cite{carla} scene. In the figures, the occluded vehicle in front and the pedestrian far to the right can only be detected by the larger model without the help of cooperative perception.}
	\label{yolox}
\end{figure}


In this work, we focus on the sensor sharing scheduling with raw-level cooperative perception, and propose an online learning-based algorithm to minimize the energy consumption.
To the best of our knowledge, this is the first paper addressing the scheduling problem of raw-level cooperative perception in vehicular networks. 
Our main contributions are summarized as follows:

1) We formulate the sensor sharing scheduling problem in cooperative perception as a variant of the Multi-Armed Bandit (MAB) problem, taking AI model performance, time-varying wireless channels, and power consumption into consideration.

2) An Adaptive Volatile Upper Confidence Bound (AVUCB) algorithm is proposed to minimize the total power consumption of the ego vehicle while satisfying delay and accuracy constraints. The upper bound of performance loss is derived.

3) Simulations are carried out for stationary and dynamic settings respectively, verifying the near-optimal performance of our proposed algorithm. Nearly $40\%$ of the energy consumption is saved compared to the random policy. 

\section{System Model} 
\subsection{System Overview}

\begin{figure}[!t]
	\centering
	\includegraphics[width=7.5cm]{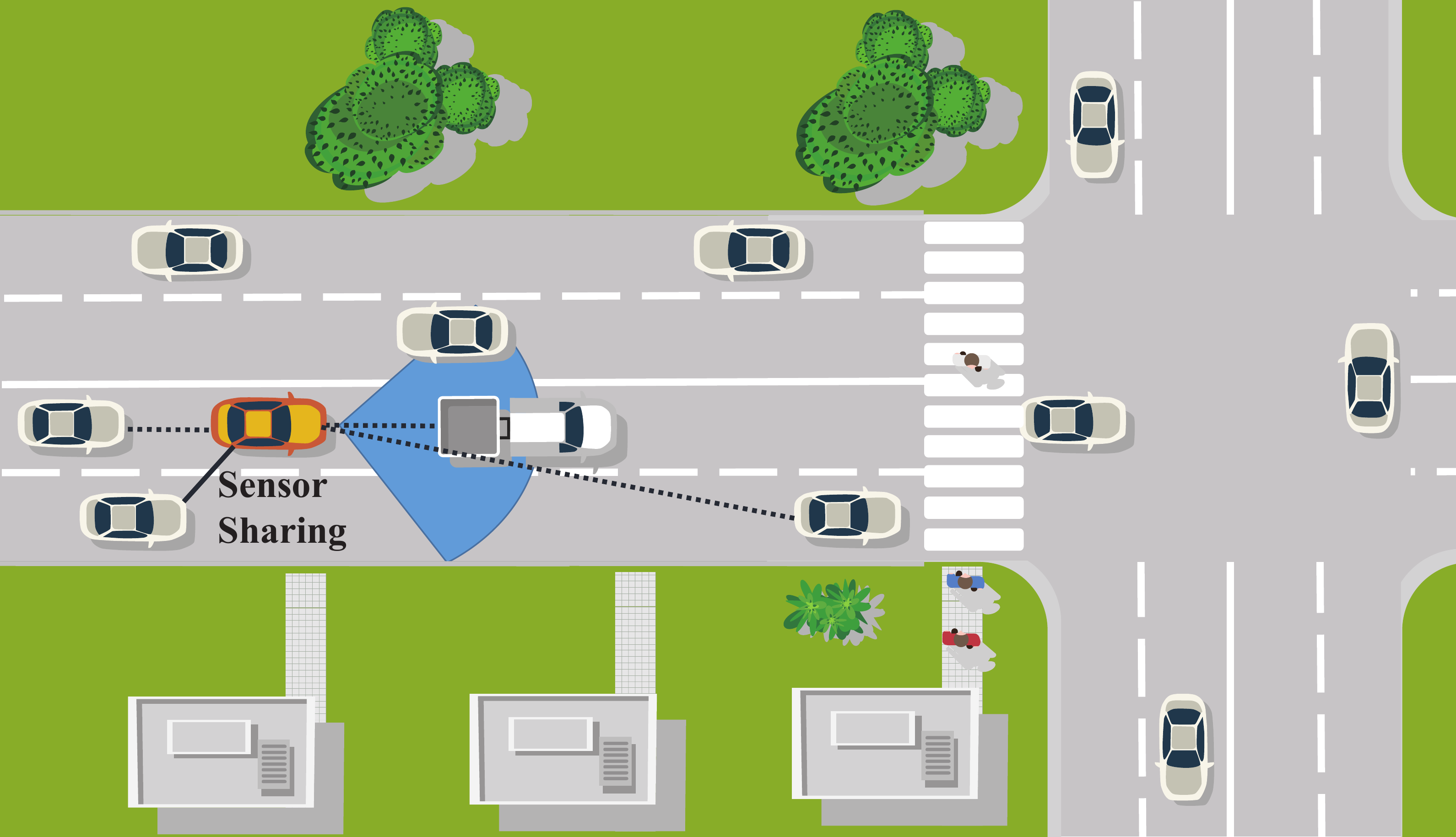}  
	\caption{An illustration of the cooperative perception. The ego vehicle (yellow car) schedules the transmission from cooperating vehicles (white cars) in the same direction to share their sensory data.}
\end{figure}

The system consists of a group of cooperating vehicles on a multi-lane street. 
A single-user scenario is considered: An L4/L5 autonomous driving vehicle needs to perceive the environment and distinguish the traffic around it using a cooperative perception framework. 
We consider a time slotted system with time index $t$ and the length of a time slot is $\tau$.
At time $t$, the nearby V2X-connected vehicle set is denoted by $\mathcal{N}(t)$.
The ego vehicle requests and receives sensor data from a vehicle $a_t\in\mathcal{N}(t)$, and then runs a neural network of dynamic complexity to merge the two sources and obtain the detection results. 
We assume that the nearby V2X-connected vehicles come and leave over time and $\mathcal{N}(t)\ne\emptyset$ for any $t$, otherwise the vehicle will perceive the environment by itself and run a relatively complex detection model to guarantee the detection performance.

The common performance metric for the object detectors is average precision (AP), measuring the accuracy of a certain neural network model when applied to a given dataset. 
By definition, AP summarizes the precision-recall curve by computing the mean precision percentage value for sweeping recall values. 
For safety consideration, the minimum AP $R_0$ should be met so that the vehicle can get enough environmental information and drive safely.

\subsection{Computation Model}
Assume that the autonomous vehicle is equipped with a smart computation hardware that can adapt its voltage and clock frequency to different workloads and save energy, by dynamic voltage and frequency scaling (DVFS) technique.
Let $L_t$ denote the computation load, measured in floating-point operations per second (FLOPS), and $f_t$ denote the clock frequency which corresponds to the computation capability. The computation time is given by
\begin{equation}
    T_t^\text{comp} = L_t/f_t.
\end{equation}

According to a common model of DVFS devices in edge computing\cite{guo2016energy}, the energy consumption is given by
\begin{equation}
    E_t^\text{comp} = \kappa L_t f_t^2 = \kappa L_t^3 / (T_t^\text{comp})^2,
\end{equation}
where $\kappa$ is the effective switched capacitance and it is hardware-specific. 
It can be seen from the equation that the energy consumption is proportional to the cubic of the workload and inversely proportional to the square of the given time.

\subsection{Communication Model}
We consider a Vehicle-to-Vehicle (V2V) wireless channel for the transmission of sensor data.
Assume that each transmission is constantly $D$ bits, since the input size is fixed given the detection model. According to Shannon's formula, the communication latency is
\begin{equation}
    T^{\text{comm}}_t(i) = \frac{D}{B\log_2\left(1+\frac{P h_t(i)}{\sigma^2}\right)},
\end{equation}
where $B$ is the bandwidth of the channel, $P$ is the transmitting power and $h_t(i)$ is the channel state between the nearby vehicle $i\in\mathcal{N}(t)$ and the ego vehicle at time $t$. The wireless channel states $h_t(i)$ are assumed to be independent over time, and their distributions are unknown to the ego vehicle. The energy required for communication is given by
\begin{equation}
    W_t^{\text{comm}}(i) = PT^{\text{comm}}_t(i).
\end{equation}


\subsection{Cooperative Perception Model}
To characterize different traffic scenarios, we define the \textit{context complexity} at time $t$ as $\omega_t$.
At crowded intersections or in bad weather, $\omega_t$ is higher since the objects are harder to identify due to more occlusions.
On the other hand, since the additional sensor data could provide a more complete view of occluded objects or better depth prediction for distant targets, it boosts the detection performance by a good margin. 
Define $\eta_{t}(i)\ge0$ as the \textit{performance gain} at time $t$ from shared view of nearby vehicle $i$. 
Then we can express the detection performance of the dynamic network as $g(L_t;\omega_t, \eta_t(a_t))$, measured by average precision. 

Since the relative positions among autonomous vehicles change at the rate of seconds, we assume that the performance gain from sensors of other vehicles $\eta_t(i)$ is quasi-stationary.
Specifically, we define an \textit{epoch} as the period of time in which the set of nearby vehicles are fixed, and the distribution of their performance gain $\eta_t(i)$ and channel state $h_t(i)$ does not change during an epoch. 
Additionally, the distribution of $\eta_t(i)$ is different across vehicles due to occlusion relationships, heterogeneity of sensor quality, and compatibility between systems. 
Note that a nearby vehicle with altered position, such as lane changes, will be re-identified as a new vehicle and triggers the start of an epoch. 

\section{Problem Formulation}
In this problem, the ego vehicle makes sequential decisions about which vehicle $a_t\in\mathcal{N}(t)$ to request for sensor sharing in order to minimize the total energy consumption subject to an accuracy constraint and a delay constraint over time $t=1,2,\dots,T$, i.e.
\begin{align}
    \min_{a_t\in\mathcal{N}(t)} & \sum_{t=1}^T E_t (a_t) = \sum_{t=1}^T \kappa L_t^3 / (T^{\text{comp}}_t)^2 + PT^{\text{comm}}_t(a_t) , \\
    \text{s.t.}  \quad & g(L_t ; \omega_t, \eta_t(a_t)) \ge R_0 ,\\
    \quad &  T^{\text{comm}}_t(a_t) + T^{\text{comp}}_t \le \tau .
\end{align}
At time $t$, the ego vehicle observes the context complexity $\omega_t$ from recent perception results, and chooses a vehicle to request sensor information. 
Upon receiving the sensor data, the dynamic neural network predicts the performance gain $\eta_t(a_t)$ and decides the computation load $L_t$ to satisfy the performance requirement.
Then the hardware runs at minimum computation frequency, for energy efficiency, to finish the task within the time slot $\tau$. 
For feasibility, we assume $T^{\text{comm}}_t(a_t) \le T^{\text{comm}}_\text{max} < \tau$ in the worst channel situation so that the inference can be completed on time.

Note that the ego vehicle does not know the channel state $h_t(i)$ and the performance gain $\eta_t(i)$ of the nearby vehicle $i\in\mathcal{N}(t)$ at the beginning of the time slot $t$.
It can only make decisions based on past information, which forms a variant of the Multi-Armed Bandit (MAB) problem\cite{slivkins2019introduction}.
Specifically, the ego vehicle makes sequential decisions on which vehicle to request for sensor sharing, observes and learns from past results to estimate the distribution of the channel state and the performance gain. 

The difference between our problem and the standard MAB problem lies in 3 aspects:
Firstly, we take into account the unreliability of wireless communications. 
Even if a nearby vehicle provides sensor data with the highest performance gain, it may not be optimal when its wireless channel is not good. 
Secondly, the context complexity $\omega_t$ is time-varying, which means the extra energy cost of not choosing the optimal vehicle varies over time. 
Thirdly, the volatility of nearby vehicles is considered, and thus the algorithm should be capable of adapting to the gradually changing environment.

\subsection{Offline Optimal Solution}
Suppose that we have access to the random distribution of performance gain and channel states of each vehicle. 
At time $t$, the remaining time for the inference task of each available nearby vehicle is given by
\begin{equation}
    \label{comp-time}
    T^{\text{comp}}_t(i) = \tau - \frac{D}{B\log\left(1+\frac{P h_t(i)}{\sigma^2}\right)}, \quad i\in\mathcal{N}(t).
\end{equation}
Further we derive the expected total energy consumption as
\begin{equation} \label{minus-square}
    \E[E_t(i)] =  \mathbb{E}\left[\frac{\kappa [g^{(-1)}(R_0;\omega_t, \eta_t(i))]^3}{T^{\text{comp}}_t(i)^2} + PT^{\text{comm}}_t(i)\right],
\end{equation}
where $g^{(-1)}(R_0;\omega_t, \eta_t(i))$ denotes the required computation load given the performance requirement. 
The expectation is taken over the probability distribution of performance gains $\eta_t(i)$ and channel states $h_t(i)$ of each vehicle, which is exactly what we aim to learn in the online scenario.
Finally, we choose the vehicle that minimizes the expected total energy consumption:
\begin{equation}
    a_t^* = \arg\min_{i\in\mathcal{N}(t)} \E[E_t(i)].
\end{equation}

\subsection{A Volatile Opportunistic MAB Problem}
To make the problem concrete, we empirically assume a logarithmic relationship between computation load and detection performance, i.e.
\begin{equation} \label{logrel}
    g(L_t ; \omega_t, \eta_t(a_t)) = m\log(1+n L_t)-\omega_t+\eta_t(a_t),
\end{equation}
where $m,n$ are parameters for the model, to be fit later in Section V. 
Note that the terms of context complexity and performance gain are directly added to the precision, representing the occluded targets and the additional information provided by the shared sensor information respectively.

Given the context complexity $\omega_t$ and the shared sensor data, the dynamic neural network adaptively decides the computation load as
\begin{equation}
    L_t \approx \frac{1}{n} e^{(R_0+\omega_t - \eta_t(a_t))/{m}},
\end{equation}
where the performance gain from another view, $\eta_t(a_t)$, is estimated by the network.

In practical settings, the energy cost of communication is usually negligible compared to computation in autonomous driving, and thus the energy consumption of detection task is
\begin{align}
    E_t(a_t) &= \frac{\kappa}{n^3} e^{\frac{3}{m} (R_0+\omega_t-\eta_t(a_t)} / \left(\tau-T_t^\text{comm}(a_t)\right)^2.
\end{align}

We define the weighting factor as $W_t=e^{\frac{3}{m}\omega_t}$, which decides the magnitude of energy consumption at time $t$, and define $X_t(a_t)=e^{-\frac{3}{m}\eta_t(a_t)} / \left(\tau-T_t^\text{comm}(a_t)\right)^2$, which relates to the wireless channel state and performance gain of the scheduled vehicle $a_t$.
Finally, the problem is translated into a combination of volatile MAB\cite{wu2018adaptive} and opportunistic MAB\cite{bnaya2013social} problem:
\begin{equation}
    \min_{a_t\in\mathcal{N}(t)}  \  \sum_{t=1}^T \E[{E_t(a_t)}] = \frac{\kappa}{n^3}e^{\frac{3}{m} R_0} \sum_{t=1}^T W_t \E\left[X_t(a_t)\right].
\end{equation}

\section{AVUCB Algorithm for Sensor Sharing}
Observe that when the context complexity $\omega_t$ is high, the scheduling decision is more important because a sub-optimal channel state and performance gain pair could lead to higher extra energy consumption. 
Intuitively, we prefer to take cautious action, to fully exploit known information rather than to explore. 

Based on the consideration above, we propose an adaptive volatile UCB (AVUCB) algorithm to balance the trade-off between exploration and exploitation in the proposed sensor sharing problem, as shown in Algorithm 1.

\begin{figure}[htbp]
    \renewcommand{\algorithmicrequire}{\textbf{Parameter:}}
    \begin{algorithm}[H]
    \caption{AVUCB Algorithm for Sensor Sharing}
    \begin{algorithmic}[1]
    	\REQUIRE $\beta$ 
    	\FOR {$t=1,\cdots,T$}
    	    \IF {Any candidate vehicle $i\in\mathcal{N}(t)$ has not shared sensor data with the ego vehicle}
                \STATE Send vehicle $i$ sharing request, i.e. $a_t=i$.
                \STATE Observe $T^\text{comm}_t(a_t)$ when receiving shared data.
                \STATE Feed the two sources to the dynamical network and complete computation within the time slot.
                \STATE Observe $\eta_t(a_t)$ from network output.
                \STATE Initialize: \ $\bar{X}_{t,a_t} = e^{ - \frac{3}{m} \eta_t(a_t)} / (\tau-T^{\text{comm}}_t(a_t))^2$, \\ $k_{t,a_t}=1,\ s_{a_t} = t$.
            \ELSE
            \STATE Observe $\omega_t$ from previous environmental perception.
            \STATE Calculate the normalized weighting factor: \begin{equation}
                \Tilde{W_t} = \max \left( \min \left( \frac{e^{\frac{3}{m}\omega_t}-e^{\frac{3}{m}\omega_L}}{e^{\frac{3}{m}\omega_H}-e^{\frac{3}{m}\omega_L}}, 1 \right), 0\right).
            \end{equation}
	        \STATE Calculate the optimistic cost function for each available nearby vehicle $i\in\mathcal{N}(t)$:
	            \begin{equation} \label{optcost}\begin{aligned}
    	            \Tilde{X}_{t,i} = 
    	            \bar{X}_{t-1,i} - \sqrt{\frac{2\beta(1-\Tilde{W_t})\log(t-s_i)}{k_{t-1,i}}}.
	            \end{aligned} \end{equation}
	        \STATE Send sensor sharing request to vehicle
	        \begin{equation}
	            a_t=\arg\min_{i\in\mathcal{N}(t)} \bar{X}_{t,i}.
	        \end{equation}
	        \STATE Observe $T^\text{comm}_t(a_t)$, run the dynamic network and observe $\eta_t(a_t)$ from network output.
	        \STATE Calculate the cost function:\begin{equation}
	            X_{t}(a_t) = e^{ - \frac{3}{m} \eta_t(a_t)} / (\tau-T^{\text{comm}}_t(a_t))^2.
	        \end{equation}
	        \STATE Update  $\bar{X}_{t,a_t} \gets \frac{\bar{X}_{t-1,a_t} k_{t-1,a_t}+X_{t}(a_t)}{k_{t-1,a_t}+1}$.
	        \STATE Update $k_{t,a_t}\gets k_{t-1,a_t}+1$.
	        \ENDIF
	   \ENDFOR
    \end{algorithmic}
    \end{algorithm}
\end{figure}
 
In lines 3-7, we explore the newly available candidate vehicle once as the initialization process, to deal with the volatility of the vehicular cluster.
Afterwards, in lines 9-12, we calculate the optimistic cost function, consisting of the empirical average cost and an exploration term, based on history observations, and choose the vehicle that minimizes it. 
Note that the exploration term is controlled by the normalized weighting factor that depends on $\omega_t$ and is truncated to satisfy $\Tilde{W_t}\in[0,1]$, incorporating context-awareness.
$\beta$ is the constant that controls level of exploration in UCB-based algorithms.
Finally, in lines 13-16, the empirical cost function and number of scheduled times are updated.

We analyze the performance of the proposed algorithm for the rest of the section. Denote the total number of epochs by $B$, the starting time and ending time of epochs by $t_b, t_b',\ b=1,\cdots,B$, then $\mathcal{N}(t) = \mathcal{N}_b$ for $t\in[t_b, t_b']$. 
Let $\mu_i=\E[X_{t,i}]$ denote the expected cost for vehicle $i\in\mathcal{N}_b$, and for optimal cooperating vehicle $a_b^*=\arg\min_{i\in\mathcal{N}_b} \mu_i$, $\mu_b^* = \min_{i\in\mathcal{N}_b} \mu_i$. 
The cumulative regret is given by
\begin{equation} \begin{aligned}
    R_t &= \E\left[ \sum_{b=1}^B \sum_{t = t_b}^{t_b'} W_t ( X_t(a_t)) - X_t(a^*_b) ) \right] \\ 
    &= \sum_{b=1}^B \sum_{t = t_b}^{t_b'} W_t \left[\mu_t(a_t)-\mu_b^*\right].
\end{aligned} \end{equation}

Let $\Delta_i=\mu_t(i)-\mu_b^*,\ t\in[t_b, t_b'],\ i\in\mathcal{N}_b$, we have the following upper bound for the cumulative regret. 
\begin{theorem}
    For $\sqrt{\beta}=1/(\tau-T^{\text{comm}}_\text{max})^2$, with random continuous $\omega_t$ and $\Pr\left\{\omega_t \le \omega_L\right\}=\rho > 0$, using AVUCB algorithm, the expected cumulative regret is upper bounded by
    \begin{equation}
        \E[R_T] \le \sum_{b=1}^B e^{\frac{3}{m}\omega_\text{max}}\left[\sum_{i\neq a_b^*} \frac{8\log(t_b'-t_b)}{\Delta_i}+O(1)\right].
    \end{equation}
\end{theorem}
    \begin{proof}
    In epoch $b$, we first derive a bound for the number of times vehicle $i\in\mathcal{N}_b$ is scheduled. Divide the optimistic cost function \eqref{optcost} by $\sqrt{\beta}$, then the cost function is normalized to $[0,1]$ because for any $t$ and vehicle $i$,
    \begin{equation}
            X_{t}(i) / \sqrt{\beta} = e^{ -\frac{3}{m} \eta_t(i)} \left(\frac{\tau-T^{\text{comm}}_\text{max}} {\tau -T_i^{\text{comm}}}\right)^2  \le 1 .
    \end{equation}
    Now the problem is equivalent to the standard adaptive MAB problem defined in \cite{wu2018adaptive} except for that we have better empirical estimation, due to the observations to vehicles that exist even before epoch $b$, which only makes the concentration property stricter. Define $\E[k_{t,i}^{(b)}]$ as the number of times vehicle $i$ is scheduled inside epoch $b$. By Lemma 5 in \cite{wu2018adaptive}, we have
    \begin{equation}
        \E[k_{t,i}^{(b)}] \le \frac{8\log(t_b'-t_b)}{\Delta_i^2} + O(1).
    \end{equation}
    Finally, we calculate the cumulative regret as
    \begin{equation} \begin{aligned}
        E[R_T] &= \sum_{b=1}^B \sum_{t = t_b}^{t_b'} W_t \E \left[  \mu_t(a_t)-\mu_b^* \right] \\
        &\le \sum_{b=1}^B e^{\frac{3}{m}\omega_\text{max}} \sum_{i\neq a_b^*} \Delta_i \E[k_{t,i}^{(b)}] \\
        &\le \sum_{b=1}^B e^{\frac{3}{m}\omega_\text{max}} \sum_{i\neq a_b^*} \Delta_i \left[\frac{8\log(t_b'-t_b)}{\Delta_i^2} + O(1)\right] \\
        &= \sum_{b=1}^B e^{\frac{3}{m}\omega_\text{max}}\left[\sum_{i\neq a_b^*} \frac{8\log(t_b'-t_b)}{\Delta_i}+O(1)\right].
    \end{aligned} \end{equation}
    \end{proof}

The theorem implies that our proposed AVUCB algorithm can achieve a logarithmic performance loss, $O(B\log(T))$ in specific, compared to the offline optimal solution which knows in prior the channel and gain distributions of all vehicles.

\section{Simulation Results}
In this section, we first specify the parameters of the model, and then evaluate the energy consumption of the proposed AVUCB algorithm through simulations of a stationary and a dynamic scenario.

\subsection{Fitting of Parameters}
We choose the MS-COCO dataset\cite{lin2014mscoco} as the baseline for context complexity to evaluate detection performance. 
We fit $m,n$ in \eqref{logrel} with YOLOX\cite{ge2021yolox} and YOLOR\cite{wang2021yoloR}, two state-of-the-art object detection models and obtain the best fit as
\begin{equation}
    g(L_t ; \omega_t, \eta_t(a_t)) = 4.695\log(1+200.9L_t)-\omega_t+\eta_t(a_t),
\end{equation}
with goodness $R^2=0.99$, as shown in Fig. \ref{fit}. 
Moreover, the effective switched capacitance is set as $\kappa = 0.98 \mathrm{W} \cdot s^{2} \cdot (\mathrm{TFLOP})^{-3}$, where $\mathrm{TFLOP}$ represents Tera  Floating-point Operations, based on the power and computation capability statistics of the Turing GPU in Nvidia Drive AGX Pegasus Platform \cite{pegasus}. 

\begin{figure}[htbp]
	\centering
	\includegraphics[width=8cm]{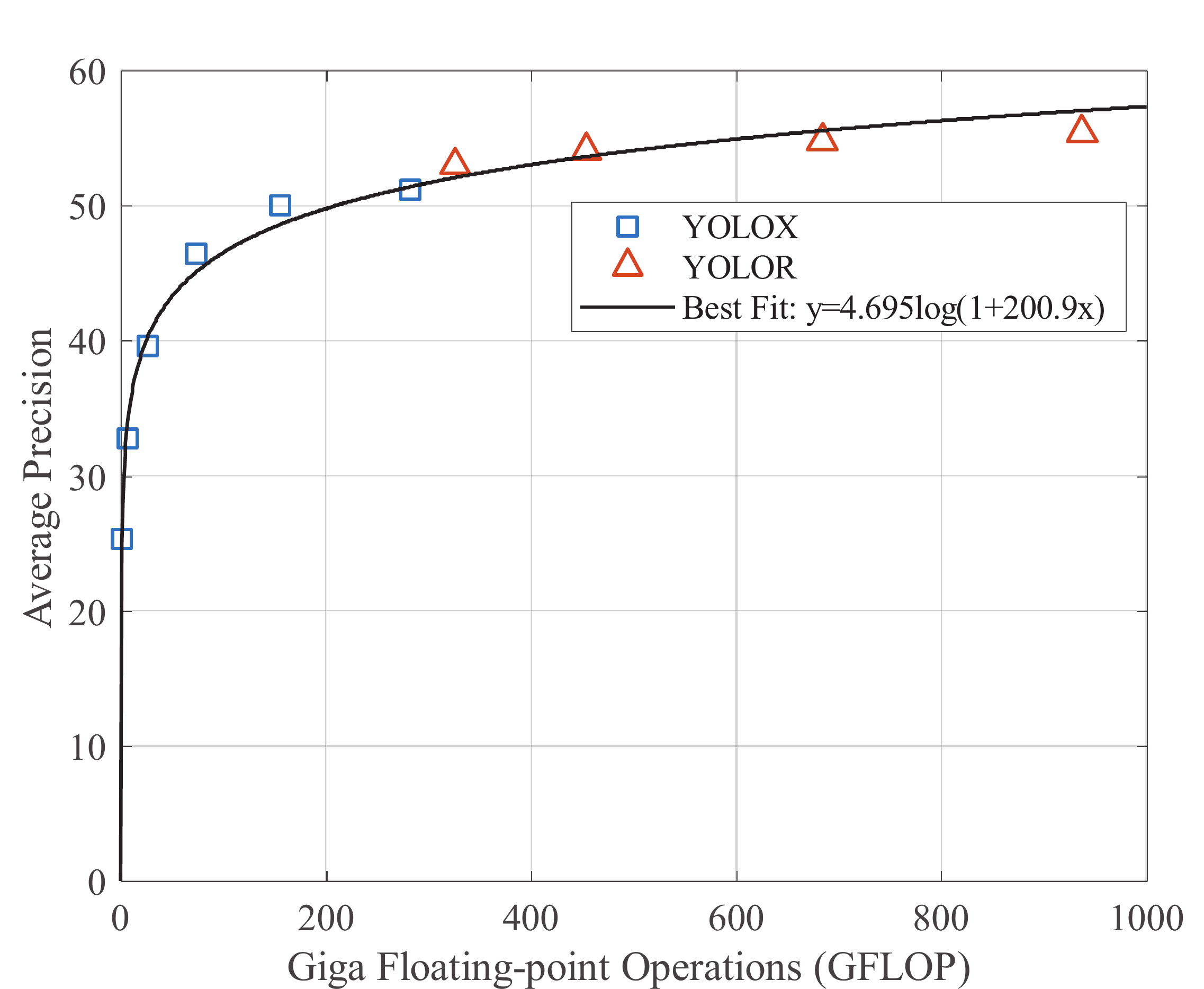}
	\caption{A logarithmic fit for computation load and detection performance based on models of YOLOX\cite{ge2021yolox} and YOLOR\cite{wang2021yoloR} with different sizes.}
	\label{fit}
\end{figure}

\subsection{Simulation under Stationary Setting}
In the simulation, we specify the minimum average precision $R_0=55$ for the baseline complexity.
In complicated traffic scenarios, such as the intersection, the context complexity is set as $\omega^{(1)}=+2$, while in simple traffic we set $\omega^{(2)}=-2$. 
We model the dynamics of the traffic scenario as a two-state Markov chain.
Suppose a complex scenario is crossed for $3.0$ seconds and a simple road is passed for $6.0$ seconds on average, then the transition interval is exponential distribution with mean values $3.0$ and $6.0$ seconds respectively. 

Similarly, we assume a two-state Markovian channel model for V2V communications, representing the Line-of-Sight (LoS) or Non-Line-of-Sight (NLoS) channel between the ego vehicle and cooperative vehicles. 
According to the V2V channel model in \cite{2016LOSboban}, the channel states of LoS and NLoS are $h_\mathrm{LoS}=-85\mathrm{dB}$, $h_\mathrm{NLoS}=-100\mathrm{dB}$ at the distance of about $100$ meters, respectively. 
The channel transition interval is again an exponential distribution, with a mean value of $1.0$ second. 
The transmitting power is $P =0.1$W, the bandwidth is specified as $B=10\mathrm{MHz}$ and the size of sensor data is $D = 2\mathrm{MB}$, the typical compressed size of a 1080p HD picture. 
The time slot is set as $\tau=50\mathrm{ms}$. Under this parameter setting, the energy consumed by communication can be neglected by orders of magnitude. 

The average performance gain of each candidate vehicle is uniformly distributed, i.e. $\eta_\mathrm{avg}(i)\sim\mathcal{U}[0,5]$, and the actual performance gain at time $t$ is given by $\eta_t(i)= \max(0, \eta_\mathrm{avg}(i) + \mathcal{N}[0, \sigma^2])$, where the randomness of the gain stems from the longitudinal relative movements and the frequent changes in occlusion relationship, and the standard deviation is specified as $\sigma = 2$.

In the stationary setting, we assume that there are 10 V2X-connected vehicles in the neighborhood of the ego vehicle, and the horizontal relative positions of nearby vehicles are fixed for $60$ seconds, when the average performance gain $\eta_{\max}(i)$ are constant. 
The proposed AVUCB algorithm is compared to several baselines: random policy, $\epsilon$-greedy algorithm, vanilla UCB algorithm and offline optimal solution. We generate $10^6$ traces for Monte Carlo experiments, and obtain the results in Fig. \ref{energy-stationary}.

\begin{figure}[!t]
	\centering
	\includegraphics[width=7.8cm]{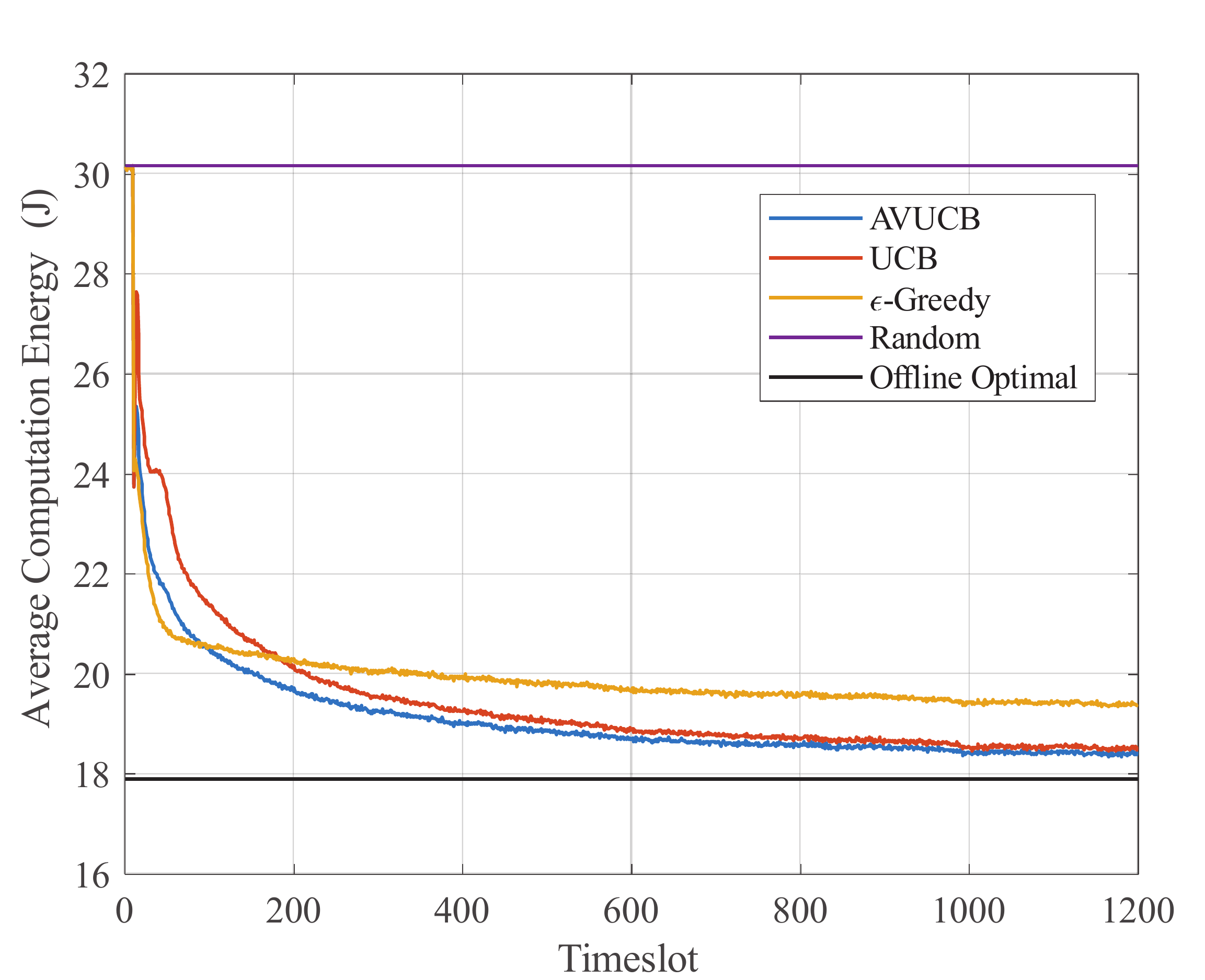}
	\caption{Comparison of average energy consumption in the stationary setting. }
	\label{energy-stationary}
\end{figure}


The results show that all learning algorithms can converge to low energy consumption after only a few hundreds of time slots and beat the random policy. Among them, the proposed AVUCB algorithm explores more cautiously than classic UCB at the early exploration stage and learns the underlying distributions better than $\epsilon$-greedy algorithm, striking a good balance between exploration and exploitation. With the proposed algorithm, the ego vehicle consumes an average power of less than $19\text{J}\times20\text{fps} = 380W$ and saves nearly $40\%$ of the energy consumption compared to the random policy. 

\subsection{Simulation under Dynamic Setting}
To investigate the influence of arrivals and departures of candidate vehicles and time-varying traffic scenarios, we compare the AVUCB algorithm with baseline algorithms in a specific dynamic setting. 
We pick a period of the synthetic scenario and simulate with the same parameters as in the stationary setting. 

\begin{figure}[!t]
	\centering
	\includegraphics[width=7.8cm]{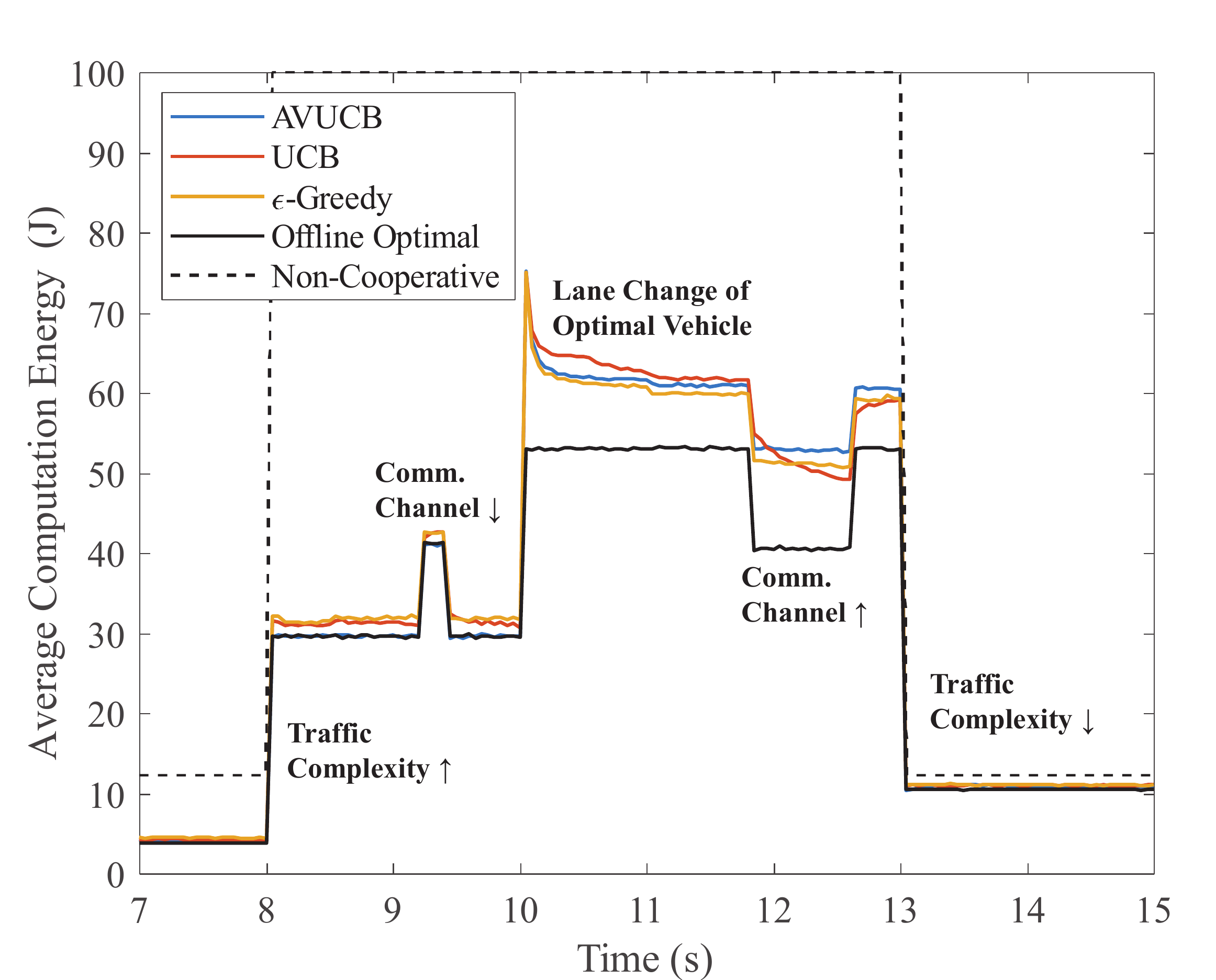}
	\caption{Comparison of average energy consumption in the dynamic setting.}
	\label{dynamic}
\end{figure}

Fig. \ref{dynamic} shows the average energy consumption in the dynamic scenario. 
It can be seen that all the learning algorithms are near-optimal after the initial learning period.
At $8$ s, the traffic complexity goes up, where AVUCB stops exploring and sticks to the empirically optimal vehicle to save energy. 
When the optimal vehicle changes lanes at $10$ s and is re-labeled as a sub-optimal vehicle, AVUCB greedily exploits the previously learned optimal vehicle to avoid risks, leaving the exploration to the future time of lower traffic complexity, i.e., after $13$ s.

\section{Conclusions}
In this work, we have studied the sensor sharing scheduling problem in raw-level cooperative perception and formulated a variant of the MAB problem with the purpose of minimizing energy consumption. We have proposed an online learning-based AVUCB algorithm that achieves logarithmic performance loss asymptotically compared to the offline optimal solution and then verified the effectiveness of the proposed algorithm with simulations. 
Future directions include considering the multi-user scenario and conducting more realistic experiments with traffic simulators such as SUMO and CARLA for a fine-grain model of detection.

\bibliographystyle{IEEEtran}
\bibliography{IEEEabrv, bibliography}

\end{document}